\documentclass[letter paper, 10 pt, conference]{ieeeconf}

\IEEEoverridecommandlockouts                              

\usepackage[utf8]{inputenc}
\usepackage{cite}
\usepackage{amsmath,amssymb,amsfonts}
\usepackage{steinmetz}
\usepackage{algorithmicx}
\usepackage{algorithm}
\usepackage{algpseudocode}
\usepackage{graphicx}
\usepackage{amsmath,amssymb}
\usepackage{textcomp}
\usepackage{algpseudocode} 
\usepackage{mathtools}
\usepackage{dsfont}
\usepackage[dvipsnames]{xcolor}
\usepackage{nicefrac}
\usepackage{txfonts}
\usepackage{tikz}
\usetikzlibrary{automata,arrows,positioning,calc}
\usepackage{hyperref}
\usepackage{cleveref}
\usepackage{todonotes}

\title{\LARGE\bf  Target Defense against Sequentially Arriving Intruders:\\ Algorithm for Agents with Dubins Dynamics}

\author{Arman Pourghorban and Dipankar Maity
\thanks{The authors are with the Department of Electrical and Computer Engineering, University of North Carolina at Charlotte, NC, 28223, USA. 
Email: 
{\tt apourgho@charlotte.edu, dmaity@charlotte.edu}
}%
\thanks{
This research is supported by the ARL grant ARL DCIST CRA W911NF-17-2-0181
}
}

\newcommand{\rt}{\rho_{_T}}
\newcommand{\ra}{\rho_{_A}}
\newcommand{\xd}{\mathbf{x}_{D}}
\newcommand{\xa}{\mathbf{x}_{A}}

\newcommand{\R}{\mathbb{R}^2}
    
    \newcommand{\x}{\textbf{x}}
    \newcommand{\ro}{r_{_T}}
    \newcommand{\p}{\textbf{p}}
\newtheorem{lemma}{Lemma}

\newtheorem{remark}{Remark}
\newtheorem{assumption}{Assumption}
\begin{document}

\maketitle
\thispagestyle{empty}
\pagestyle{empty}
\begin{abstract}

We consider a variant of the target defense problem where a single defender is tasked to capture a sequence of incoming intruders. Both the defender and the intruders have non-holonomic dynamics.
The intruders' objective is to breach the target perimeter without being captured by the defender, while the defender's goal is to capture as many intruders as possible. After one intruder breaches or is captured, the next appears randomly on a fixed circle surrounding the target. 
Therefore, the defender's final position in one game becomes its starting position for the next.
We divide an intruder-defender engagement into two phases, \textit{partial information} and \textit{full information}, depending on the information available to the players.  We address
the capturability of an intruder by the defender using the notions of \textit{Dubins path} and \textit{guarding arc}.
We quantify the percentage of capture for
both finite and infinite sequences of incoming intruders. Finally, the theoretical results are verified through numerical examples using Monte-Carlo-type random trials of
experiments.

\end{abstract}


\section{Introduction}
Considering the growing capabilities of robotics for securing and supervising regions, considerable research has focused on guarding a target \cite{obstaclePEG,sensingPEG,Bhattacharya2010}. We study a perimeter-defense game between a defender and a team of intruders. The intruder team sends its members sequentially to breach the target, while the defender is tasked with guarding the target by intercepting them. 

Perimeter defense games, first introduced in \cite{isaacs1999differential}, are a variant of pursuit-evasion problems in which a defender aims to capture intruders before they breach a protected region \cite{velhal2022decentralized}. These problems arise naturally in applications such as coastline defense \cite{coastline_defense}, area patrolling \cite{area_monitoring}, area securing \cite{area_securing}, and border protection \cite{MultiAgentPerimeterDefenseReview}. Several works have investigated perimeter defense where the defender is restricted to the boundary of the target \cite{Guerrero2021, lee2021defendingperimetergroundintruder}, as well as more general cases where the defender can move freely in the environment \cite{bajaj2022competitiveperimeterdefenseconical, reach-avoidYan, maity2024cooperative}. Much of this literature has considered ``one-shot'' formulations, where all intruders are present from the start and the game concludes once each intruder has either been captured or reached the target. 

In contrast, sequential-arrival formulations \cite{bajaj, bajaj2022competitive,pourghorban2023target} account for intruders that arrive over time, producing a repeated series of defender-intruder engagements. These works, however, almost exclusively rely on simplified point-mass dynamics, which allow instantaneous heading changes and yield symmetric, circular reachability sets (i.e., the Apollonius circle \cite{dorothy2024one}). Such assumptions facilitate analysis but overlook fundamental kinematic constraints present in real robotic platforms. In particular, ground vehicles and aerial robots exhibit non-holonomic motion with bounded curvature, meaning that their reachable sets, dominance regions, and engagement outcomes differ qualitatively from holonomic models.

Moreover, many existing studies \cite{bajaj,reach-avoidYan} assume intruders have no sensing capability, thereby limiting their ability to adapt their strategies on-the-fly. In contrast, our prior investigations \cite{pourghorban2022target,pourghorban2023targetspie} demonstrated that sensing fundamentally changes intruder behavior and defender-intruder interactions. Building on this perspective, the present work examines sequential perimeter defense where both intruders and the defender have \textit{limited sensing regions} and \textit{non-holonomic dynamics}. This combination introduces new strategic and analytical challenges, as the defender must simultaneously reason about sensing asymmetries, curvature-constrained motion, and the long-term consequences of sequential arrivals.
\begin{figure}
    \centering
    \includegraphics[trim = 0 100 0 0, clip, width = 0.3 \textwidth]{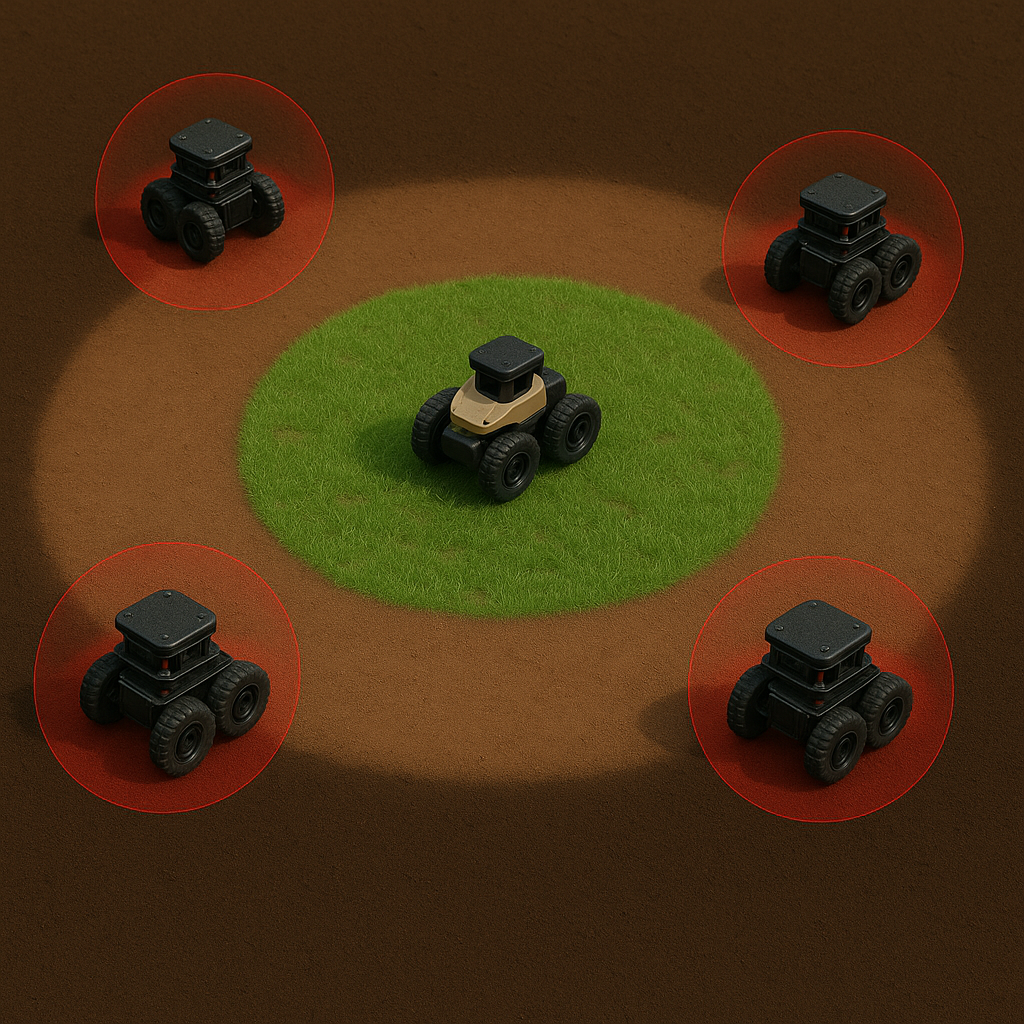}
    \put (-148, 70 ) {\line(1,0) {31}} 
    \put (-135, 75) {$\rt$}
    \caption{\colorbox{green!30}{\textcolor{green!30}{gh}}: Target region with radius $\ro$, \colorbox{red!20}{\textcolor{red!20}{gh}}: intruders sensing region with radius $\ra$, \colorbox{orange!60!black!25}{\textcolor{orange!60!black!25}{gh}}: Target's sensing anulus with radius $\rt$. \textit{This figure was generated using AI-assisted image synthesis (OpenAI DALL·E).}}
    \label{fig:introPic}
    \vspace{-6 pt}
\end{figure}
\par
Sequential intrusion is not merely a modeling convenience but reflects both practical and theoretical defense contexts. In real-world security scenarios, adversaries often stagger their intrusion attempts to reduce the chance of detection, rather than coordinating large-scale simultaneous breaches; such sequential tactics have been documented, for example, in studies of smuggling and cross-border infiltration \cite{dandurand2020migrant}. From a theoretical standpoint, sequential attack models are also well established in the security-games literature, where adversaries launch repeated or staged attacks on defended perimeters and defenders must adapt their strategies accordingly \cite{de2018facing}. This perspective highlights that sequential intrusion is both practically motivated and analytically distinct: each engagement’s terminal configuration directly shapes the initial conditions of the next, requiring defensive strategies that account not only for immediate captures but also for long-term positioning across a sequence of arrivals.
\par

Intruders arrive randomly on a circle surrounding the target and attempt to breach the target sequentially, with only one intruder active at any given time (see Fig.~\ref{fig:introPic}). Unlike holonomic settings, non-holonomic dynamics reshape the structure of engagements: they alter the feasible configurations for capture, constrain the set of favorable strategies, and introduce new geometric phenomena such as the emergence of capture circles. These features make the defender’s long-term strategy fundamentally different from point-mass formulations.

Due to the sequential arrival of intruders, the game becomes a series of \textit{one-on-one engagements} where the next begins immediately after the previous ends. We analyze each game by decomposing it into two phases: a \textit{partial information phase}, where only one agent can sense the other, and a \textit{full information phase}, where both agents are mutually aware. The defender’s (intruder’s) objective is to maximize (minimize) the percentage of captured intruders over the sequence.

The main contributions of this paper are:  
(i) Formulating and solving a sensing-limited target defense game against sequentially incoming intruders with non-holonomic dynamics,  
(ii) Characterizing intruder \textit{capturability} using Dubins-path reachability and introducing the notions of \textit{guarding arcs} and \textit{capture circles},  
(iii) Analytically computing capture percentages for both finite and infinite sequences of intruder arrivals using a Markov chain abstraction, and  
(iv) Validating the theoretical results through Monte Carlo simulations.  

The rest of the paper is organized as follows: We formulate the problem in \Cref{sec:ProbFormulation}, and in \Cref{sec:background} we outline the key assumptions and present background material. The two phases of the one-vs-one games—\textit{partial information} and \textit{full information}—are discussed in \Cref{sec:partInfo} and \ref{sec:fullInfo}, respectively. The analysis of the entire sequence of arrivals, including capture percentages for finite and infinite cases, is presented in \Cref{sec:GameAnalysis}. Numerical results are discussed in \Cref{sec:Simu}, and we conclude in \Cref{sec:Conclusion}.


\section{Problem Formulation} \label{sec:ProbFormulation}

 We consider a target guarding problem in ${\mathbb{R}^2}$ where a single defender is tasked to protect a circular target region $\mathbf{R}_{T}=\{\x \in {\mathbb{R}^2}\ | \ \|\x\| \leq r_{_T}\}$ from an incoming \textit{sequence} of intruders, as schematically shown in Fig.~\ref{fig:introPic}.
 The target is equipped with a sensing annulus of radius $\rt$, hereafter referred to as the \textit{Target Sensing Region} (TSR). 
 The TSR enables the defender to detect the presence of an intruder within this region.
The intruder is equipped with its own sensing capability, allowing it to sense its surrounding within a distance of $\ra$ (c.f., Fig.~\ref{fig:introPic}).

 The intruders appear \textit{sequentially} on the boundary of the TSR. 
 Each one-on-one game between the defender and an intruder either ends in a \textit{breach} of the target or in the \textit{capture} of the intruder.
 The next intruder appears on the TSR boundary immediately after the current game ends. 
Each intruder  appears \textit{randomly} on the TSR boundary, with the arrival process having a \textit{uniform probability} over the TSR boundary and being \textit{independent} of the previous arrivals.
Due to the randomness of intruder arrivals, the number of captures (or equivalently, the capture percentage) becomes a random variable.
Our objective in this work is to synthesize a defender strategy to maximize the expected capture percentage.

Let $\xa(t),\xd(t) \in \R$ denote the positions of an intruder and the defender at time $t$. 
Let $\psi_A(t)$ and $\psi_D(t)$ denote their heading angles at that time. 
The joint state of the agents are denoted by $\xi_A = [\xa^\intercal, \psi_A]^{^\intercal}$ and $\xi_D = [\xd^\intercal, \psi_D]^{^\intercal}$, which follow the dynamics
\begin{align} \label{eq:dynamics}
    \dot\xi_A = \begin{bmatrix}
        \nu_A \cos(\psi_A)\\
        \nu_A \sin(\psi_A)\\
        \omega_A
    \end{bmatrix}, \qquad  
    \dot\xi_D = \begin{bmatrix}
        \nu_D \cos(\psi_D)\\
        \nu_D \sin(\psi_D)\\
        \omega_D
    \end{bmatrix},
\end{align}
where $\nu_D$ and  $\omega_D$  ($\nu_A$ and $\omega_A$) are the defender's (intruder's) linear and angular velocities, respectively.
In this problem, we assume that the defender and the intruder move with the maximum linear speeds of $1$ and $\nu$, respectively, meaning, $|v_D(t)| = 1$ and $|v_A(t)| = \nu$ for all $t$. 
Additionally, we assume that $\nu < 1$, indicating that the defender is faster. 
Here $\omega_A$ (similarly, $\omega_D$) denotes the control action of the intruder (similarly, defender). 

\par
Upon appearing on the TSR boundary, the intruder moves radially toward the target center with maximum speed until it detects the defender.  The defender captures the intruder if $\|\xd(t)-\xa(t) \|=0$ and $\| \xa(t)\|> r_{_T}$. However, the intruder breaches the target if $\|\xa(t) \|<r_{_T}$ and $\|\xa(t)-\xd(t) \|>0$.
The intruder being equipped with its own sensor can detect the defender only if the defender is within a distance of $\ra$ or less. 
Leveraging this sensing capability, the intruder can identify the optimal breaching point on the target, escape the TSR without being captured, or evade the defender for a period before eventual capture. This evasive maneuver is a crucial capability for the intruder, as it compels the defender to pursue and capture the intruder in a location that is disadvantageous for the defender to begin the next game. Thus, while getting captured, each intruder can maximize the likelihood of winning for the next intruder, which would not have been possible if $\ra =0$.

\section{Background and Assumptions} \label{sec:background}
\subsection{Dubins Path}

 \begin{figure}[t]
     \centering
    { \includegraphics[trim = 207 280 190 275, clip,width=0.23\textwidth]{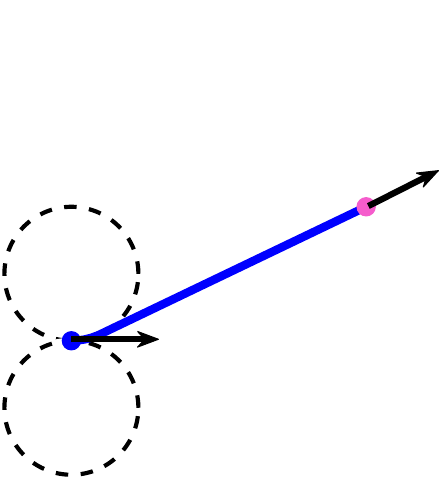}}
    \put (-108,27) {$\x_0, \psi_0$}
    \put (-24,66) {$\x_f$}
     \centering
    { \includegraphics[trim = 207 280 190 275, clip,width=0.23\textwidth]{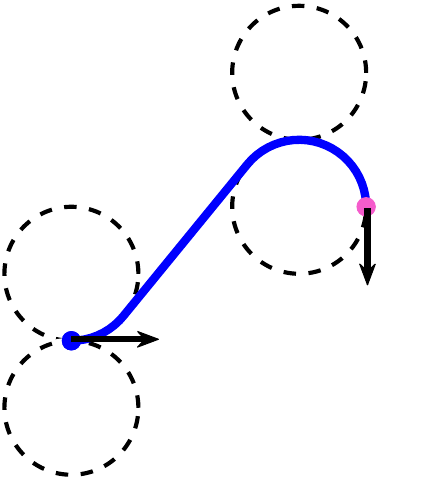}}
    \put (-108,27) {$\x_0, \psi_0$}
    \put (-50,73) {$\x_f, \psi_f$}
      \caption{Left: Dubins path when the final orientation $\psi_f$ is not specified. Right: Dubins path when the final orientation $\psi_f$ is specified.
      The arrows indicate the agent's initial and final heading angles. The dashed circles are the turning circles with radius $R$ representing the minimum radius of curvature that the agent must maintain during its turns.} 
      \label{fig:dubins_path}
    \end{figure}
The Dubins path is a mathematical model that finds the shortest path between two points in a plane while imposing constraints on the curvature and requiring continuous forward motion.
Formally, a Dubins path is defined as the shortest path between two configurations $\xi_0=[\x_0^\intercal, \psi_0]^{^\intercal}$ and $\xi_f=[\x_f^\intercal, \psi_f]^{^\intercal}$. Given a non-holonomic  vehicle constrained by a turning radius $R=\nicefrac{\nu}{\omega}$ the Dubins path is composed of segments from the set $\{C,S\}$, where $C$ represents a circular arc trajectory segment with constant radius $R$ and $S$ represents a straight line trajectory segment; see   \cite{meyer2015dubins,buzikov2021time,li2022auction} for more details. 

The agent's path depends on the relative positions and orientations of the start and goal configurations. 
When the final heading angle $\psi_f$ is specified, the path consists of three components: a circular arc, a straight-line segment, and a final circular arc. 
However, if the final orientation $\psi_f$ is not specified, meaning the agent is free to arrive at the destination with any heading, the path consists of only two components: a circular arc and a straight-line segment, see Fig.~\ref{fig:dubins_path}.

\subsection{Finite Time Reachability and Dominance Region}
Given an initial configuration  ${\xi_A}_0 = [ {\xa}_0^\intercal,{\psi_A}_0]^\intercal$  of the intruder, we compute the set of points that the intruder can reach in a finite time $T$. This \textit{reachable set}  is denoted by $\mathcal{R}({\xi_A}_0,T)$:
\begin{align}
    \mathcal{R}({\xi_A}_0,T)= \{\x \in \R:  &~~ \exists ~u_A(\cdot) \text{ and } \exists ~t \in [0,T] \nonumber \\
    &\quad\text{s.t. } \ {\xi_A}(0)={\xi_A}_0,~~~ \xa(t)=\x \},
\end{align}
where $u_A(.)$ is an admissible control input for the intruder. 
This reachability set can be efficiently computed using existing tools such as \cite{bansal2017hamilton}. 
In Fig.~\ref{fig:reachable_set}, we show a few examples of such reachable sets by varying the linear and angular speeds. 

When the players can instantly change their heading (i.e., the control variable is $\psi$ instead of $\omega$), such as the scenarios in \cite{english2021defender,manoharan2023nonlinear,lai2021reach}, the reachable set $\mathcal{R}({\xi_A}_0,T)$ becomes a circle with center ${\xa}_0$ and radius $\nu T$, where $\nu$ is the maximum allowed linear speed of the intruder.
 \begin{figure}
     \centering
    { \includegraphics[trim = 205 315 158 270, clip,width=0.2\textwidth]{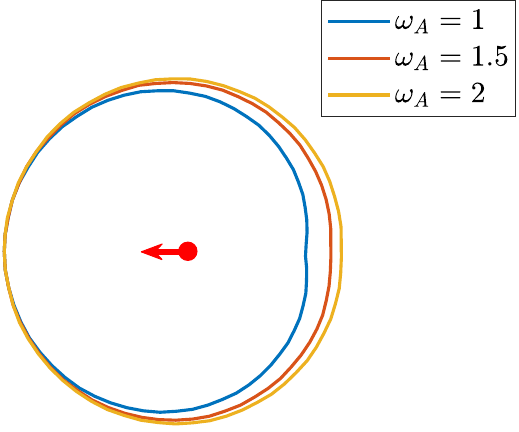}}
     \centering
    { \includegraphics[trim = 205 315 158 270, clip,width=0.22\textwidth]{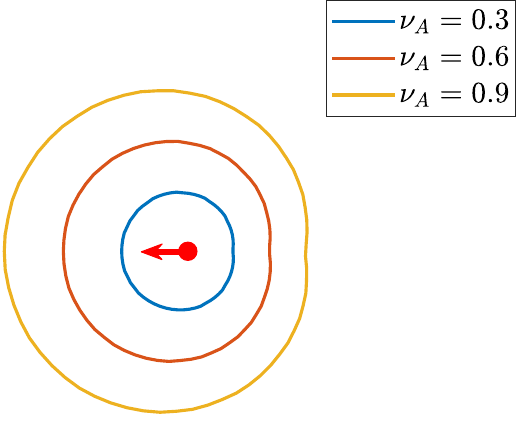}}
      \caption{The reachable set $\mathcal{R}(\xa,\psi_A,T)$ for $T=5$ is presented. Left: For $\nu_A=0.9$, $\mathcal{R}({\xa}_0,{\psi_A}_0,T)$ is shown by varying $\omega_A=1, 1.5$ and $2$. Higher value of $\omega_A$ corresponds to a larger region.  Right: For $\omega_A=1$, $\mathcal{R}({\xa}_0,{\psi_A}_0,T)$ is shown by varying $\nu_A=0.3, 0.6$ and $0.9$.}
      \label{fig:reachable_set}
    \end{figure}

While $\mathcal{R}({\xi_A}_0,T)$ provides the reachable set of the intruder, all such points may not be reachable in presence of the defender. 
Given initial configurations ${\xi_A}=[\xa^{^\intercal}, \psi_A]^{^\intercal}$ and $\xi_D=[\xd^{^\intercal},  \psi_D]^{^\intercal}$ of the agents, we now compute the set of points that the intruder can reach before the defender, i.e., the  \textit{dominance region} of the intruder: 
\begin{align}
    \mathcal{C}(\xi_A,\xi_D)= \{\x \in \mathbb{R}^2: \exists \  u_A(.) \ \text{s.t. } \forall \ u_D(\cdot) \ t_A(\x)\le t_D(\x)\},
\end{align}
where $t_A(\x)$ and $t_D(\x)$ are the minimum time for the intruder and defender to reach the point $\x$:
\begin{align*}
    t_i(\x) = \inf \{ t ~:~  \x \in \mathcal{R}(\xi_i, t)\}, \qquad i =A,D.
\end{align*}
In the special case where the agents can instantaneously change their heading angles (i.e., $\psi$ is the control variable instead of $\omega$), the dominance region $ \mathcal{C}(\xi_A,\xi_D)$ becomes a circle with radius $\gamma\|\xa -\xd\|$ and center at $\alpha \xa - \beta \xd$, where $\alpha = \nicefrac{1}{1-\nu^2}$, $\beta = \nu^2 \alpha$ and $\gamma = \nu\alpha$.
This circle is known as the \textit{Apollonius circle} \cite{dorothy2024one}.

\subsection{Head-On Engagement}

As the name suggests, a head-on engagement happens when the agents are heading directly to each other. 
More formally, given the configurations $\xi_D$ and $\xi_A$, a head-on engagement requires $\psi_D = \angle \hat{\x}_{DA} = \psi_A \pm \pi $, where we define $\hat{\x}_{DA} := \nicefrac{(\xa-\xd)}{\| \xa-\xd\|}$ to be the unit relative position vector, and, for any vector $u \in \R$, we define $\angle u = \arctan2([u]_2, $ $ [u]_1 )$, with $[u]_j$ being the $j$-th component of $u$.

\subsection{Parametric Assumptions} \label{sec:assumptions}
\begin{assumption}
    Parameters $\nu, \omega_A,\omega_D, \ra,\rt,r_{T}$ ensure that there exists configurations $\xi_A$ and $\xi_D$ such that $\|\xa-\xd\| = \ra$ and $\mathcal{C}(\xi_A,\xi_D) \cap \mathbf{R}_T = \emptyset$ and $\max_{\x \in \mathcal{C}(\xi_A,\xi_D)} \|\x\| \le \rt + \ro - \ra. $ 
\end{assumption}

This assumption is necessary to ensure that there exists at least one strategy for the defender to capture the intruder inside the TSR.
Without this assumption, every intruder will either be able to breach the target or escape out of the TSR.
This assumption, specifically the $\max_{\x \in \mathcal{C}(\xi_A,\xi_D)} \|\x\| \le \rt + \ro - \ra$ part, also prevents a deadlock situation where the defender is sensed by an intruder located outside the TSR.

\begin{assumption} 
    Let $\hat{\rho}_{_T} :=\rt + \ro - \ra$ and $m:= \nicefrac{2\hat{\rho}_{_T}\omega_D}{(\hat{\rho}_{_T}^2-1)}$. The parameters must satisfy
    \begin{align*} 
    \hat{\rho}_{_T} + \frac{\pi+\arctan(m)}{\omega_D} \le \frac{\rt}{\nu}.
\end{align*}
\end{assumption}

This assumption is necessary to ensure that the defender can return to the target center from the Capture Circle (defined later in \Cref{sec:fullInfo}) before an intruder reaches the target from the TSR boundary. 
Essentially, it guarantees a \textit{reset} of the game: if the defender cannot intercept the current intruder---which, as we will explain later, the defender knows immediately upon the intruder's appearance on the TSR boundary---they return to the target center and allow the current intruder to breach the target. The next game then begins with the defender positioned at the target center and a new intruder appearing randomly on the TSR boundary.
While this assumption is not necessary for the problem to be well-posed, it does simplify the analysis of an otherwise complex problem.

\subsection{Game Phases}
Each one-on-one game between the defender and an intruder consists of two phases \cite{shishika2021partial, pourghorban2022target}, namely, the \textit{Full Information} and the \textit{Partial Information Phase}.
\par
\textit{Full information phase}: In this phase, both agents can sense each other, i.e, $\|\xa(t)\|<\ro+\rt$ and $\|\xa(t)-\xd(t)\|\le \ra$.
\par
\textit{Partial information phase}: In general, there are two distinct ways this phase can occur. First, when only the defender sees the intruder, i.e., $\|\xa(t)\|<\ro+\rt$ and $\|\xa(t)-\xd(t)\|>\ra$. 
The second case is when the defender comes within the
sensing region of the intruder but the intruder is outside the
TSR. However, based on the assumption described in Section~\ref{sec:assumptions}, the second case will cease to occur.

In the subsequent sections, we will study both the agents' (defender and intruder) strategies in each phase of the game.
\section{Partial Information Phase} \label{sec:partInfo}
In this phase, only the defender sees the intruder and not vice-versa.
We assume that, during this phase, the intruder moves radially toward the target center at maximum speed $\nu$ until it detects the defender, at which point the full information phase begins.
The objective of the defender in this phase is to utilize its sensing advantage to initiate the full information phase in a favorable configuration.
The following discussion will address what constitutes a favorable configuration for the defender and whether such a configuration is attainable.

\begin{figure}
     \centering
    \begin{minipage}[b]{0.21\textwidth} \notag
        \centering
        \fbox{\includegraphics[trim = 200 300 190 270, clip,width=\textwidth]{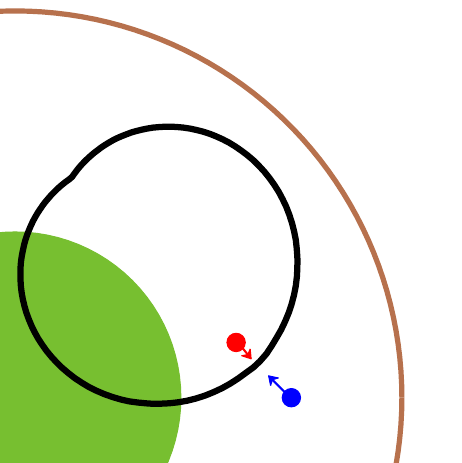} 
                    \put (-58,98) {${\rm Breach \ only}$}
        }
        \label{fig:image1}
    \end{minipage}
    \centering
    \hspace{0.05 cm}
    \begin{minipage}[b]{0.21\textwidth}
        \centering
        \fbox{\includegraphics[trim = 200 300 190 270, clip,width=\textwidth]{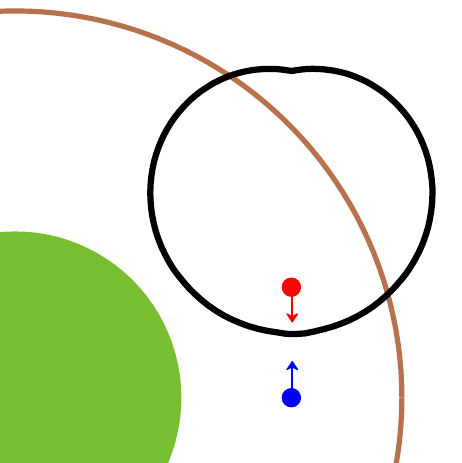}
        \put (-58,98) {${\rm Escape \ only}$}
        } 
        \label{fig:image2}
    \end{minipage}
    \begin{minipage}[b]{0.21\textwidth}
        \centering
        \fbox{\includegraphics[trim = 200 300 190 270, clip,width=\textwidth]{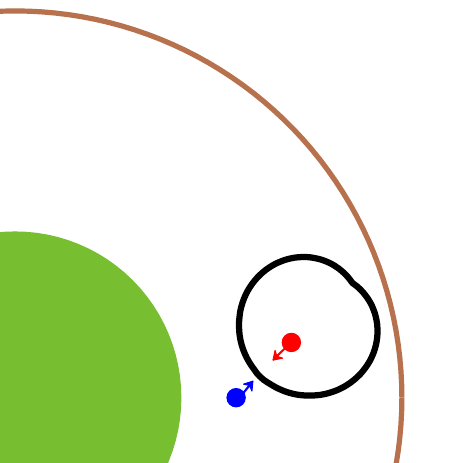} 
        \put (-58,98) {${\rm Capture \ only}$}
        }
        \label{fig:image3}
    \end{minipage}
    \centering
       \hspace{0.05 cm}
    \begin{minipage}[b]{0.21\textwidth}
        \centering
        \fbox{\includegraphics[trim = 200 300 190 270, clip,width=\textwidth]{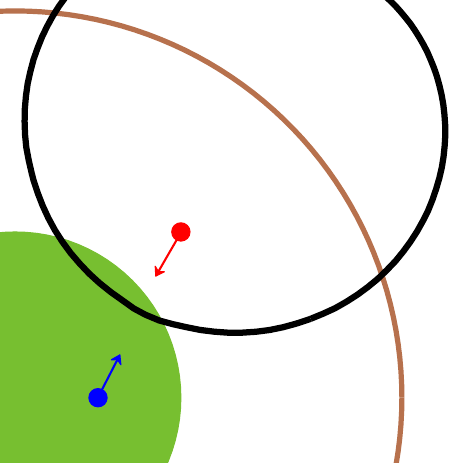}
        \put (-90,98) {${\rm Breach \ or \ Escape}$}
        }
        \label{fig:image4}
    \end{minipage}

    \caption{The blue and red dots represent the positions of the defender and
the intruder, respectively. The black line represents  $\mathcal{C}(\xi_A,\xi_D)$. 
The green region is a part of the circular target
and the brown line represents the target sensing boundary. 
}
    \label{fig:main_scenarios}
\end{figure}

\subsection{Defender Strategy}\label{sec:defender_strategy}
A similar problem involving first-order dynamics (instead of the non-holonomic dynamics \eqref{eq:dynamics}) requires the defender to engage with the intruder at the \textit{right} time and in the \textit{right} place, which necessitates a strategic waiting period by the defender \cite{shishika2021partial}. 
Loosely speaking, this strategy ensures that the initial engagement (i.e., the start of the \textit{full information phase}) does not occur too close to the TSR boundary---where the intruder could escape---or too late in the game, where the resulting configuration leads to an inevitable breach. See Fig.~\ref{fig:main_scenarios} for an illustration of four possible outcomes.

\begin{figure}
    \centering
    \includegraphics[trim = 160 280 145 320, clip, width = 0.4 \textwidth]{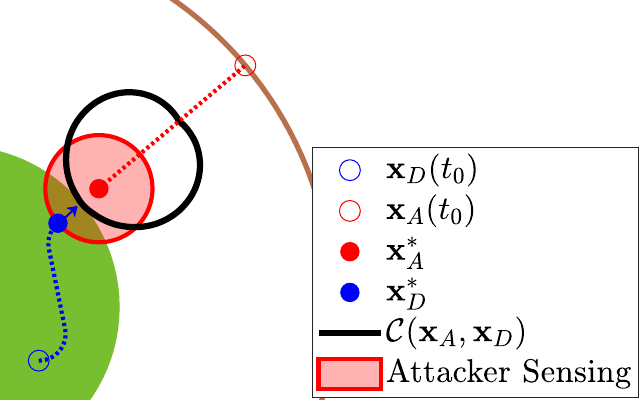}
    \caption{Critical configuration at time $t_A^*$.}
    \label{fig:defender_strategy}
    \vspace{-6 pt}
\end{figure}

Let $r_D^* $ be the smallest radial position of the defender such that $\mathcal{C}(\xi_A,\xi_D)$ does not intersect with the target when a head-on engagement occurs between the defender and the intruder, see Fig.~\ref{fig:defender_strategy}.
More formally, consider an intruder appearing on the TSR boundary at an angle $\varphi$ which detects the defender in a head-on engagement configuration at time~$t$. 
At the moment of detection, the intruder will be at the location $r_A(t)[\cos\varphi, \sin\varphi]^\intercal$, where $r_A(t) = \ro+\rt - \nu t$.\footnote{
Recall that the intruder moves radially toward the target center until it detects the intruder. Therefore, $\psi_A=\varphi \pm \pi$.
} 
The head-on engagement along with the intruder being able to detect the defender requires the defender to be exactly at the location $\xd(t) = (r_A(t)-\ra)[\cos\varphi, \sin\varphi]^\intercal$, with heading angle $\psi_D(t) =\varphi$. Consequently, the defender's radial location is $r_D(t)=\|\xd(t)\|=\ro+\rt-\ra-\nu t$. 
The dominance region of the intruder at this configuration, $\mathcal{C}(\xi_A(t), \xi_D(t))$, depends on only on $t$, since all the other quantities the fixed for a given $\varphi$. 
The smallest radial local $r_D^*$ is defined to be the solution to the following optimization problem:
\begin{align} \label{eq:optimization}
\begin{split}
    \inf_t~\qquad &r_D(t) \\
    \text{subject to   }\quad & \mathcal{C}(\xi_A(t), \xi_D(t)) \cap \mathbf{R}_T = \emptyset, \\
    & r_D(t) = \ro+\rt-\ra - \nu t.
\end{split}    
\end{align}

\begin{remark}
    The set of all feasible solutions to \eqref{eq:optimization}, or equivalently, the corresponding set of radial locations $r_D$'s, characterizes all possible head-on engagement locations. 
    It is noteworthy that the shape and size of $\mathcal{C}(\xi_A(t), \xi_D(t))$ remain unchanged as $t$ varies, except that the set shifts closer to the target center as $t$ increases.
    When the engagement occurs at the optimal time $t^*$, the capture takes place closer to the target center. 
    This proximity increases the likelihood of a successful capture in the next round \cite{pourghorban2022target}, which motivates the use of $t^*$ in designing the defender strategy as discussed later.
    Although it may be possible to engage earlier than $t^*$, such engagements may not be optimal for the remainder of the incoming sequence. This leads to the emergence of a strategic waiting behavior---similar to what was previously observed in \cite{shishika2021partial} and \cite{pourghorban2022target}.  \hfill $\blacksquare$
\end{remark}

\begin{remark}
    Because of the circular game region, we note that $r_D^*:=r_D(t^*)$ is independent of the intruder's arrival angle $\varphi$; therefore, the above optimization needs to be solved only once, rather than each time a new intruder appears. 
    This significantly reduces the computation.~\hfill$\blacksquare$
\end{remark}

As soon as an intruder appears on~the TSR, the defender computes the engagement location as $\xd^* = r_D^*[\cos\varphi, \sin\varphi]^\intercal$. 
Using the Dubins shortest path, the defender also calculates the time required to reach the head-on engagement configuration $\xi_D^* = [{\xd^*}^\intercal, \varphi]^\intercal$. 
If this time exceeds the time the intruder takes to reach the corresponding engagement location $\xa^* = (r_D^* + \ra)[\cos\varphi, \sin\varphi]^\intercal$, which is $t_A^* = \frac{\ro+\rt - r_D^* -\ra}{\nu}$, the defender chooses not to engage and allows the intruder to breach the target.
In that case, the defender repositions itself to maximize the capture probability for the next intruder.
Otherwise, the defender follows the Dubins shortest path to reach the configuration $\xi_D^*$; see the blue dotted path in Fig.~\ref{fig:defender_strategy}.

Now, given a defender configuration $\xi_D$ and an intruder arrival angle $\varphi$, we let $\tau(\xi_D, \varphi)$ denote the time taken by the defender to reach the  engagement configuration $\xi_D^*$. 
Let us further define the set of arrival angles $\varphi$ for which the defender can reach the engagement condition as follows:
\begin{align*}
    \theta_{G}(\xi_D) = \{ \varphi~|~ \tau(\xi_D, \varphi) \le \frac{\ro+\rt - \ra -r_D^*}{\nu}~~\}.
\end{align*}
This set, $\theta_G$, is referred to as the \textit{guarding arc}, as any intruder appearing on the TSR with an angle in this set will be captured---as will be shown in the next section---while any intruder appearing from outside this set will successfully breach the target. 
An illustration of the guarding arc corresponding to a defender configuration $\xi_D(t)$ at time $t$ is shown as the magenta arc  in Fig.~\ref{fig:final_capture} (subplot on the right).

To summarize the defender's strategy: at a given time $t$, when a new intruder appears on the TSR boundary, the defender immediately determines whether it can reach $\xi_D^*$---in which case the full information phase will occur---or whether it should return to the target center to maximize the probability of capturing the next intruder. 
Consequently, the defender either follows the Dubins path to the configuration $\xi_D^*$, or returns to the target center with an unconstrained final heading. 
The final heading is left unconstrained at the target center because the next intruder appears on the TSR boundary with uniform probability, making all defender headings equally effective.

\section{Full Information Phase} \label{sec:fullInfo}
As discussed in the previous section, if an intruder appears from the set  $\theta_G$, full information phase will occur. 
In this section, we describe the tactical motion strategies for both agents during the full information phase.
 \begin{figure}
     \centering
         { \includegraphics[trim = 207 290 187 255, clip,width=0.23\textwidth]{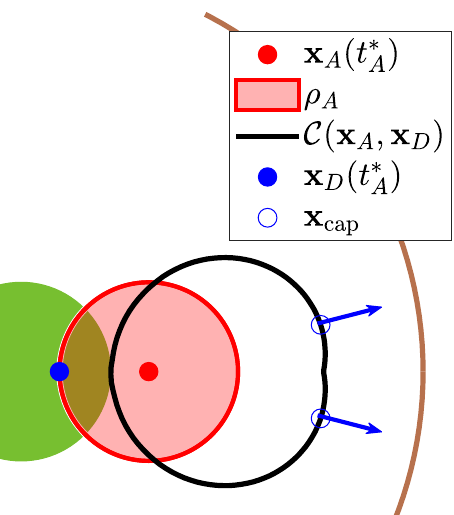}}
     \centering
    { \includegraphics[trim = 207 275 175 275, clip,width=0.23\textwidth]{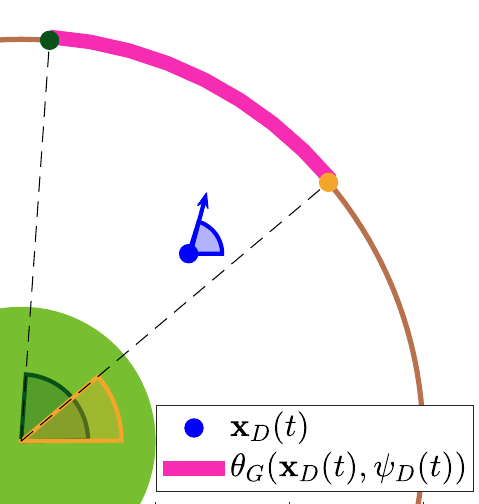}}
          \put (-95,55) {$\xd(t)$}
       \put (-65,72) {$\psi_D(t)$}
      \caption{Left:At time $t_A^*$, the set $\mathcal{C}(\xi_A,\xi_D)$ is shown. Capture location $\x_{\rm{cap}}$ and corresponding heading angle $\psi_{\rm{cap}}$ are shown with blue color.  Right: For a blue point defender $\xd(t)$, the
guarding arc $\theta_G(\xd(t), \psi_D)$ is shown with pink color. The angles $L(\xd(t), \psi_D(t))$ and $U(\xd(t), \psi_D(t))$ are shown with yellow and dark green colors, respectively. } 
      \label{fig:final_capture}
    \end{figure}
\subsection{Agent Strategies}
Recall that the defender only engages with an intruder only if it can guarantee a capture.
This is due to the constraint $\mathcal{C}(\xi_A(t), \xi_D(t)) \cap \mathbf{R}_T = \emptyset$ in \eqref{eq:optimization}. 
Given the sequential nature of the game, the intruder's objective is to minimize the defender's probability of capturing the next intruder.
Recall that the intruders' arrival follows a uniform and independent stochastic process. 
Therefore, the probability of capturing the next intruder is $|\theta_G(\x_{\rm{cap}},\psi_{\rm{cap}})|$, where $|\cdot|$ denotes the length of the interval, and $\x_{\rm{cap}} \in \mathcal{C}(\xi_A,\xi_D)$  is the location where the current intruder will get captured and $\psi_{\rm{cap}}$ is the heading angle of the defender at the capture time, see Fig.~\ref{fig:final_capture} (left).

Recall from the definition of $\mathcal{C}(\xi_A,\xi_D)$ that, for every point $\x$ on the boundary of $\mathcal{C}(\xi_A,\xi_D)$, there exists a Dubins path for the defender (and the intruder) that reaches the point $\x$ in the shortest time and the two agents meet. 
The final heading angle of the defender depends on the point $\x$, and hereafter denoted as $\psi(\x)$. 
The intruder's objective can be expressed as
\begin{align} \label{eq:intruder_optimization}
    \min_{\x \in \partial\mathcal{C}(\xi_A^*,\xi_D^*)} |\theta_G(\x, \psi(\x))|, 
\end{align}
and the minimizer of this optimization is the optimal capture point $\x_{\rm{cap}}$. 
Due to the symmetry of the shape $\mathcal{C}(\xi_A,\xi_D)$, there might be two solution for $\x_{\rm{cap}}$, see Fig.~\ref{fig:final_capture}. 
In this work, we assume that both players agree on a point. 
For practical purposes, one may consider that the defender will wait for an (infinitesimal small) amount of time to observe whether the intruder is moving clockwise or counter-clockwise to determine which of the two $\x_{\rm{cap}}$ locations is chosen by the intruder.


Due to the circular symmetry of the problem, one may verify that $\| \x_{{\rm cap}}\|$ is independent of the intruder arrival angle $\varphi$, and that $|\psi(\x_{\rm cap}) - \varphi|$ is a constant that is independent of $\varphi$ and only depends on the problem parameters. 
The circle concentric with the target with radius $\| \x_{{\rm cap}}\|$ is hereafter referred to as the \textit{Capture Circle} since all the captures occur on this circle. 
This is because, starting from this circle, the defender will again perform a head-on engagement with the next incoming intruder and thus will again end up on this circle---a phenomenon that has been observed in other optimal  defender strategies for sequential arrival \cite{pourghorban2022target}.

Starting from the capture circle, if the defender is unable to start a successful head-on capture, it moves to the target center. 
Then, next time, it attempts a head-on engagement again with the next intruder. 
If successful, the defender will again end up on the capture circle. 

\begin{remark}
    Notice that a natural min-max optimization arises in this setting: the intruder seeks to minimize the capture probability for the next intruder by solving the minimization problem in \eqref{eq:intruder_optimization}, while the defender aims to maximize it by selecting the engagement point $\xi_D^*$ (or equivalently, $t^*$ from \eqref{eq:optimization}).  
    Once the full information phase begins, the final outcome is determined by the intruder, since $\mathcal{C}(\xi_A,\xi_D)$ represents the intruder’s dominance region, and it can reach any point within this set without being captured. 
    On the other hand, the defender leverages its information advantage in the partial information phase to pick the optimal engagement configuration for the full information phase.~\hfill $\blacksquare$ 
\end{remark}


\section{Analysis of the Game} \label{sec:GameAnalysis}
The first game starts with the defender being at the target center. The defender is able to capture the intruders arriving from angle $\varphi$ if $\varphi \in \theta_G(\xi_D(0))$. Otherwise, the defender remains at the target center untill the next intruder appears. 
If the defender captures an intruder it will be located on the \textit{Capture Circle}.
At the end of each game, it will either be on the capture circle or at the target center depending on the arrival angles of the intruders.  The defender strategy (and the progression of the game) is described by Algorithm~\ref{euclid}.

\begin{algorithm} 
    \caption{Defender's Strategy}
    \label{euclid}
    \begin{algorithmic}[1] 
    \State Initialize $\xd \gets [0,0]^\intercal$, $\psi_D \gets 0$, $N_{\rm capture}\gets 0$,  and $N$
    \For{$n = 1: N$}
    \State $\varphi \sim {\mathcal U}(-\pi, \pi)$ 
    \Comment{Uniform random arrival of intruder}
    \State $\psi_A \gets \varphi + \pi$
    \If{$\varphi \in \theta_G(\xi_D)$} \Comment{{\color{blue!60}Capture happens}}
    \State $\xd^* \gets r_D^*[\cos\varphi, \sin\varphi]^\intercal,\qquad \psi_D^* \gets \varphi$
    \State Defender follows Dubins path to reach $[{\xd^*}^\intercal, \psi_D^*]$
    \State $(\x_{{\rm cap}}, \psi(\x_{{\rm cap}})) \gets $ solution to \eqref{eq:intruder_optimization}. 
    \State $\xi_D \gets (\x_{{\rm cap}}, \psi(\x_{{\rm cap}}))$   \Comment{$\xi_D$ after capture}
    \State $N_{\rm capture} \gets N_{\rm capture} +1$
    \Else \Comment{{\color{red!60} Breach happens}}
    \State Defender goes to/stays at the target center 
    \State $\xd \gets [0,0]^\intercal$
    \EndIf
    \EndFor
    \end{algorithmic}
\end{algorithm}
\subsection{Probability of a Single Capture}
As discussed earlier, a defender staring from the configuration $\xi_D=[\xd^\intercal, \psi_D]^\intercal$ can only capture intruders appearing from angles in the set $\theta_G(\xi_D)$. 
Due to the circular symmetry of the problem, one may verify that $\theta_G$ only depends on $\|\xd\|$ and $\psi_D - \angle\xd$; see Fig.~\ref{fig:final_capture} (subplot on the right).
For a defender located on the target center, one may further verify---using the circular symmetry---that the size of the set $\theta_G$ (i.e., $|\theta_G|$) is independent of the defender's heading angle $\psi_D$. 
Consequently, for a defender located on the target center, we denote the capture probability as $p_1^* = \frac{|\theta_G([0,0]^\intercal,0)|}{2\pi}$. 

Now, consider a defender who has just captured an intruder and is starting the next game.
Recall that after capturing an intruder, the defender ends up in the configuration $[\x_{\rm cap}^\intercal, \psi(\x_{\rm cap})]^\intercal$. 
One may verify---again using the circular symmetry---that $\angle\x_{\rm cap} - \psi(\x_{\rm cap})$ is constant and independent of both $\x_{\rm cap}$ and the arrival angle of the intruder that was captured.
Furthermore, as discussed in \Cref{sec:fullInfo}, that $\|\x_{\rm cap}\|$ is also constant and independent of the intruder's arrival angle; in fact, it defines the radius of the {capture circle}.  
Putting all of this together, we observe that $|\theta_G(\|\x_{\rm cap}\|,~\angle\x_{\rm cap} - \psi(\x_{\rm cap}) )|$  is independent of the arrival angle of the captured intruder and is therefore constant.
Consequently, the capture probability for a defender who has just captured an intruder becomes $p_2^* = \frac{|\theta_G(\|\x_{\rm cap}\|,~\angle\x_{\rm cap} - \psi(\x_{\rm cap}) )|}{2\pi}$.
Formally, we define\footnote{With slight abuse of notation, we write $\theta_G(\xi_D)=\theta_G(\|\xd\|, \psi_D-\angle\xd)$.} 
\begin{subequations} \label{eq:p*}
   \begin{align} 
   &p_1^* \triangleq  \frac{|\theta_G(0,0)|}{2\pi}, \\
   &p_2^* \triangleq \frac{|~\theta_G(r_{\rm cap},~\theta_{\rm cap})~|}{2\pi},
\end{align} 
\end{subequations}
where $r_{\rm cap}$ is the radius of the capture circle and $\theta_{\rm cap} \triangleq  \angle\x_{\rm cap} - \psi(\x_{\rm cap})$.

In summary, if a defender is at the target center, it can capture the next intruder with probability $p_1^*$; if it is on the Capture circle, the capture probability becomes $p_2^*$.

\subsection{Capture Percentage for a Sequence of Incoming intruders}
Since the defender starts the game from the target center, the first intruder gets captured with probability $p_1^*$. Therefore, after the first game, the defender is on the \textit{capture circle} with probability $p_1^*$ or remains at the target center with probability $1- p_1^*$.
Starting on the \textit{capture circle}, the defender remains on the \textit{capture circle} with probability $p_2^*$ by capturing the next intruder. 
Otherwise, with probability  $1-p_2^*$, the intruder breaches the target and the defender moves to the target center to \textit{restart} the game. 
This problem can be abstracted into a two-state Markov chain (see, Fig.~\ref{fig:markovfig}) with state $S_1$ denoting  the defender being at the target center and $S_2$ denoting the defender being on the \textit{capture circle}. 
This Markov-chain abstraction aids the computation of the capture percentages. 

\begin{figure}
\centering
	\begin{tikzpicture}[->, >=stealth', auto, semithick, node distance=3cm]
	\tikzstyle{every state}=[fill=white,draw=black,thick,text=black,scale=1]
	\node[state]    (A)                     {$S_1$};
	\node[state]    (C)[right of=A]   {$S_2$};
	\path
	(A) edge[loop left]			node{$p_1^*$}	(A)
  edge[bend left,above]	node{$1-p_1^*$}	(C)
	(C) edge[bend left,below]		node{$1-p_2^*$}	(A)
	 edge[loop right]		node{$p_2^*$}	(C);
	\end{tikzpicture}
    \caption{The Markov Chain of the defender's state in the game and transition between the states are presented.}
    \label{fig:markovfig}
\end{figure}
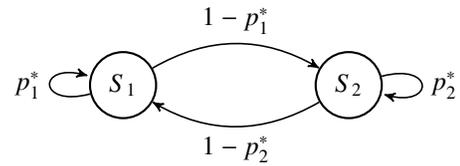
\begin{lemma}
    The expected capture percentage at the end of the $n$-th game is
    \begin{align}
    \label{eq:percentagen}
        {\rm percentage}(n)= \frac{\sum_{i=1}^n \p^\intercal \eta_i}{n}\times 100,
    \end{align}
    where $\eta_i = [\eta_i(1), \eta_i(2)]^\intercal$ with $\eta_i(j)$ denoting the probability that the state of the Markov chain is at $S_j$ at the end of the $i$-th 1-vs-1 game.
\end{lemma}
\begin{proof}
    Let $c_n$ denote the total number of captures by the end of the $n$-th game. Let us further define the random variable $s_n \in \{S_1, S_2\}$ to denote the state of the Markov chain  at time $n$; see Fig.~\ref{fig:markovfig}. 
    Let us define the binary random variables $\mu_1$ and $\mu_2$ such that 
    \begin{align*}
        \mu_i = \begin{cases}
            1, \qquad &\text{Attacker is captured from the game state } S_i,\\
            0, &\text{Attacker is not captured }
        \end{cases}
    \end{align*}
    Therefore, we may write
    \begin{align}
    \label{eq:w_n}
          c_{n+1} = c_n + \mu_1 \mathds{1}_{S_1}(s_n) + \mu_2 \mathds{1}_{S_2}(s_n), 
    \end{align}
    where $\mathds{1}_s(\cdot)$ is an indicator function such that $\mathds{1}_s(s') = 1$ if and only if $s=s'$; otherwise, $\mathds{1}_s(s') = 0$.

    Notice that, from the Markov chain in Fig.~\ref{fig:markovfig}, we have $\mathbb P(\mu_i = 1) = p_i^*$. 
    Furthermore, since the attackers appear independently of the game state state, $\mu_i$ and $s_n$ are independent random variables. 
    Taking expectations on both sides of \eqref{eq:w_n}, we obtain
    \begin{align} \label{eq:expected_w_n}
    \varmathbb{E}[c_{n+1}] = \varmathbb{E}[c_n] + p_1^* \eta_{n}(1) +  p_2^* \eta_n(2),
\end{align}
where $\eta_n(j)$ denotes the probability that the state of the Markov chain $(s_n)$ is at $S_j$ at the end of the $n$-th game. By defining the vectors $\p=[p_1^*, p_2^*]^\intercal$ and $\eta_n=[\eta_n(1),\eta_n(2)]^\intercal$, we may rewrite \eqref{eq:expected_w_n} as
\begin{align}
    \varmathbb{E}[c_{n+1}] = \varmathbb{E}[c_n] + \p^\intercal \eta_n=\varmathbb{E}[c_0]+\sum_{i=1}^n \p^\intercal \eta_i = \!\sum_{i=1}^n \p^\intercal \eta_i,
\end{align}
where $c_0=0$. Therefore, the expected percentage of capture at the end of the $n$-th game is,
    \begin{align*}
    \label{percentagen}
        {\rm percentage}(n)=\frac{\varmathbb{E}[c_n]}{n}\times 100= \frac{\sum_{i=1}^n \p^\intercal \eta_i}{n}\times 100.
    \end{align*}
    This completes the proof.
 \end{proof}
 \par 
 The asymptotic capture percentage, when $ n \rightarrow \infty$ is
\begin{align}
    {\rm percentage}(\infty) = \p^\intercal \eta^* \times 100,
\end{align}
where $\eta^*$ is the stationary distribution of the Markov Chain.

\begin{remark}
    Using the Theory of Markov Chain one may verify that 
    \begin{align*}
        \eta^* =\begin{bmatrix}
\frac{1-p_2^*}{1+p_1^*-p_2^*} & \frac{p_1^*}{1+p_1^*-p_2^*}  
\end{bmatrix}^\intercal, \ \ \eta_i =\begin{bmatrix}
1-p_1^* & 1-p_2^*  \\
p_1^* & p_2^* 
\end{bmatrix}^{i} \begin{bmatrix}
1 \\ 0  
\end{bmatrix},~~\forall ~ i. 
    \end{align*}
\end{remark}

\vspace{0.5 cm}
\section{Simulation Results} \label{sec:Simu}
We simulate the game with the following parameters $\ro =10$, $\ra = 3,~ \rt = 20$, $\nu = 0.8 $, $\omega_D=0.5$, and $\omega_A=1.5$. We considered $3$ intruder arrivals in this experiment. The first game starts when then the defender is located in the target center and $\psi_D=0$. The first intruder appears on the TSR boundary from $\varphi=1.1$rad.
The defender engages with this intruder at $\xd^*=[3.77, 7.42]$ and captures it at time $t=30.45 s$ at $\x_{{\rm cap}}=[5.06,14.58]$ with heading angle $\psi_D=1.41$ radian. 
The second intruder appears from  $\varphi=0.2$.
The defender engages with the second intruder at $\xd^*=[8.18, 1.65]$ and captures it at time $t=60.90 s$ at $\x_{{\rm cap}}=[14.57,5.09]$ with heading angle $\psi_D=0.51$ radian. The third intruder appears from  $\varphi=-2$. 
The defender is unable to capture this intruder and consequently moves toward the target center. The third intruder breaches the target at $[-4.16,-9.09]$.

\subsection{Monte-Carlo Experiment}
In this experiment we observe the outcome for a sequence of 200 arrivals.
Furthermore, as the arrival patter is random, we conducted $100$ random trials of the game. 
The game parameters remain the same as before.
The percentage of capture for each trial is plotted in Fig.~\ref{fig:percentageTrial} using a gray line.
The abscissa in this figure denotes the number of arrivals ($n$) and the ordinate denotes the percentage of capture for that number of arrivals.
We compute the empirical mean of the percentage of capture from these 100 random trials, and that is shown by the  magenta line in Fig.~\ref{fig:percentageTrial}.
To compare this simulation result with our theoretical analysis, we plot the expected percentage (i.e., ${\rm percentage}(n) $ from \eqref{eq:percentagen}) in the same figure using the 
cyan line. %
We observe that the empirical mean is very close to the theoretically predicted quantity.
In this plot, we also report the value of the asymptotic capture percentage (i.e., ${\rm percentage}(\infty)$), and we notice that the random trials and ${\rm percentage}(n)$ converge ($\sim$~exponentially) to ${\rm percentage}(\infty)$ as $n$ increases.
\begin{figure}
    \centering
    \includegraphics[trim = 200 292 180 280, clip, width = 0.3 \textwidth]{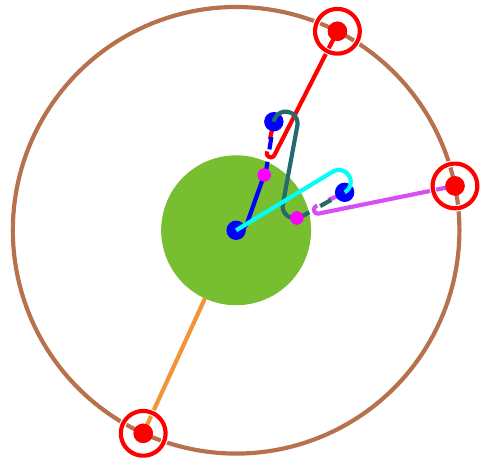}
    \put (-100,77) {\footnotesize{$t_0=0$}}
    \put (-35,133) {\footnotesize{$t_0=0$}}
    \put (-106,87) {\footnotesize{$t_A^*=23.33$}}
    \put (-109,104) {\footnotesize{$t_{\rm cap}=30.45$}}
    \put (2,84) {\footnotesize{$t_0=30.45$}}
    \put (-68,67) {\footnotesize{$t_A^*=53.78$}}
        \put (-51,94) {\footnotesize{$t_{\rm cap}=60.90$}}
        \put (-150,5) {\footnotesize{$t_0=60.90$}}
    \caption{The trajectories of the players for 3 intruder arrivals are presented. The trajectory of the defender is shown with blue, teal, and cyan lines for the arrivals, respectively. The trajectories of the intruders are shown with red, purple, and orange lines for the arrivals, respectively.}
    \label{fig:Simulation}
\end{figure}
\begin{figure}
    \centering
    \includegraphics[trim = 100 241 125 255, clip, width = 0.41 \textwidth]{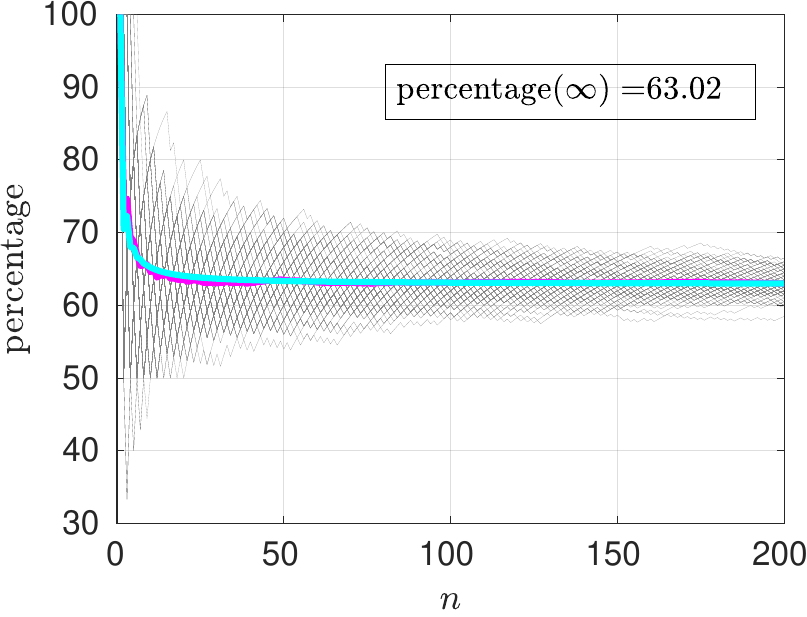}
    \caption{Percentage of Capture versus number of arrivals from 100 trials. Each trial is represented with a gray line. 
    Their empirical average is represented by the magenta line.
    The cyan line represents the theoretically predicted capture percentage in \eqref{eq:percentagen}.}
    \label{fig:percentageTrial}
\end{figure}
\section{Conclusion} \label{sec:Conclusion}
In this paper, we formulated a target defense game against
a sequence of incoming intruders with Dubins dynamics. 
The Intruders are tasked
to move radially toward the target center and breach the
target while the defender is tasked to capture as many
intruders as possible.  Based on the initial configuration of the agents and the information
available to them, each game can be divided into two phases and we discussed the agents' strategies in each phase using the notions of \textit{Dubnis path} and \textit{Guarding arc}.  We analytically computed the capture percentage for both for finite and infinite sequences of incoming intruders by abstracting the game into a two-state Markov chain.
Numerical studies are presented that corroborate the theoretical findings. 

\bibliographystyle{IEEEtran}
\bibliography{reference}

\end{document}